\documentclass[a4paper,12pt]{article}

\usepackage{amscd}
\usepackage{amsfonts}
\usepackage{amsmath}
\usepackage{amssymb}
\usepackage{amsthm}
\usepackage{color} 
\usepackage[T1]{fontenc} 
\usepackage{here} 
\usepackage{mathrsfs} 
\usepackage{txfonts} 
\usepackage[all]{xy} 
\usepackage{algorithm} 
\usepackage{algpseudocode} 
\usepackage{listings} 

\allowdisplaybreaks

\theoremstyle{plain}
\newtheorem{thm}{Theorem}[section]

\theoremstyle{definition}

\newcommand{\vs}[1][0.2]{\vspace{#1in}\noindent\ignorespaces}
\newcommand{\ba}{\begin{array*}}
\newcommand{\ea}{\end{array*}}
\newcommand{\be}{\begin{eqnarray*}}
\newcommand{\ee}{\end{eqnarray*}}
\newcommand{\bi}{\begin{itemize}}
\newcommand{\ei}{\end{itemize}}
\newcommand{\bb}{\vs\begin{itembox}}
\newcommand{\eb}{\end{itembox}}
\newcommand{\bc}{\begin{center}}
\newcommand{\ec}{\end{center}}
\newcommand{\bs}{\vs\begin{screen}}
\newcommand{\es}{\end{screen}}

\def\ens#1{{\mathchoice{\left\{ #1 \right\}}{\{ #1 \}}{\{ #1 \}}{\{ #1 \}}}}
\def\set#1#2{{\mathchoice{\left\{ #1 \middle| #2 \right\}}{\{ #1 \mid #2 \}}{\{ #1 \mid #2 \}}{\{ #1 \mid #2 \}}}}
\def\r#1{\text{\rm #1}}

\def\Bigv#1{\left| #1 \right|}
\def\v#1{{\mathchoice{\Bigv{#1}}{| #1 |}{| #1 |}{| #1 |}}}

\def\ol#1{\overline{#1}{}}
\def\tl#1{\tilde{#1}{}}

\newcommand{\bN}{\mathbb{N}}

\newcommand{\bR}{\mathbb{R}}

\newcommand{\bZ}{\mathbb{Z}}

\newcommand{\cP}{\mathscr{P}}

\newcommand{\N}{\bN}

\newcommand{\R}{\bR}
\newcommand{\Z}{\bZ}

\newcommand{\Fp}{\mathbb{F}_p}

\newcommand{\Qp}{\mathbb{Q}_p}

\newcommand{\Zp}{\mathbb{Z}_p}

\algnewcommand\algorithmicbreak{{\bf break}}
\algnewcommand\Break{\algorithmicbreak{}}
\algnewcommand\algorithmiccontinue{{\bf continue}}
\algnewcommand\Continue{\algorithmiccontinue{}}

\newcommand{\fn}[2]{
  \footnote[0]{
    $
    \begin{array}{l}
      \r{MSC2020: #1} \\
      \r{Key words: #2}
    \end{array}
    $
  }
}

\title{$p$-adic Linear Regression for Random Sampling with Digitwise Noise}
\author{Tomoki Mihara}
\date{}

\begin{document}

\maketitle
\begin{abstract}
We propose a new probabilistic algorithm of $p$-adic linear regression for random sampling with digitwise noise. This includes a new probabilistic algorithm of modulo $p$ linear regression.
\end{abstract}

\tableofcontents
\fn{65F10}{$p$-adic number, optimisation, linear regression}

\section{Introduction}
\label{Introduction}

Let $p$ be a prime number. The notion of $p$-adic numbers is introduced by K.\ Hensel in 1897 in \cite{Hen97}, and plays a central role in modern number theory. Recently, the $p$-adic numbers are used also in other branches of science, because of many significant similarities to and differences from the real numbers. Application of the $p$-adic numbers also appears in computer science. For example, S.\ Albeverio, A.\ Khrennikov, and B.\ Tirrozi studied $p$-adic neural network in \cite{AKT99} and \cite{KT00}, P.\ E.\ Bradley studied dendrograms and clusterings using $p$-adic numbers in \cite{Bra08} and \cite{Bra09}, and so on. Introduction of \cite{Bra25} explains the history well. Especially, $p$-adic equations and $p$-adic optimisation related to $p$-adic regression and $p$-adic neural networks are recently studied as a frontier topic (cf.\ \cite{ZZ23}, \cite{ZZB24}, \cite{BMP25}, \cite{Zub25-1}, \cite{Zub25-2}, \cite{Ngu25}, \cite{Mih26-1}, \cite{Mih26-2}, and \cite{Mih26-3}).

\vs
See Introduction of \cite{Mih26-1} for the similarity and the difference of real optimisation and $p$-adic optimisation. For the reader's convenience, we give a duplicated explanation here. To begin with, linear algebraic methods independent of the coefficient field and combinatorial methods such as greedy algorithm, hill climbing, and simulated annealing work also for $p$-adic optimisation. The reader should be careful that since the non-Archimedean property ensures that the sum of small error terms is still small, it prevents reaching a far point by repetition of $p$-adically small steps. Therefore, optimisation based on repetition of small steps sometimes requires to use another distance than the $p$-adic distance itself such as a variant of Hamming distance on $p$-adic expansions.

\vs
Unfortunately, there are only few success in the study of $p$-adic counterparts of real methods based on gradients. As an exception, Newton's method also works for finding a zero of a multivariable polynomial function over the field $\Qp$ of $p$-adic numbers. However, it is not applicable to an optimisation problem of a general $p$-adic function.

\vs
In addition, we cannot directly consider differentials of a loss function $\epsilon$ of the form $X \to \R$ for a subset $X$ of $\Qp$, since elements in the domain $X$ and the codomain $\R$ do not share a common arithmetic structure. Even if we extend the notion of a differential so that it vanishes at any point where $\epsilon$ is locally constant, such a formulation is not useful in $p$-adic optimisation because unlike the real setting, $\epsilon$ can be locally constant at almost all points of the domain.

\vs
In the real setting, we can use the least squares method in linear regression. However, it does not practically work the $p$-adic setting. In order to explain the difference, we recall the reason why the least squares method is effective in the real setting. Suppose that we are given a sequence $\vec{x} = (x_i)_{i=0}^{N-1}$ of sample points of the domain of an unknown function $f$ of length $N \in \N$ and the corresponding sequence $\vec{y} = (y_i)_{i=0}^{N-1}$ of observed values of $f$. When we try to express $f$ as the sum $g + \epsilon$ of a good function $g$ such as a linear function and an error function $\epsilon$, then the least squares method is an optimisation method to find such $g$ and $\epsilon$ minimising the non-negative real amount
\be
\sum_{i=0}^{N-1} \v{\epsilon(x_i)}^2.
\ee
In the real setting, the smallness of the amount ensures the smallness of each error term $\v{\epsilon(x_i)}$ with $i \in \N \cap [0,N)$, and we can express the amount as
\be
\sum_{i=0}^{N-1} (f(x_i) - g(x_i))^2
\ee
without using the absolute value. Since this expression is differentiable with respect to parameters of $g$ in many settings, optimisation methods based on differential works.

\vs
On the other hand, in the $p$-adic setting, we cannot remove the absolute value function even if we square the values, because the minimisation of
\be
\v{\sum_{i=0}^{N-1} \epsilon(x_i)^2}
\ee
is not equivalent to the minimisation of
\be
\sum_{i=0}^{N-1} \v{\epsilon(x_i)}^2.
\ee
Even if we replace the square by other powers, the problem will not be solved. Therefore, when we deal with $p$-adic optimisation, we cannot concentrate on computation of a sum of $p$-adic numbers and need to consider a sum of real numbers as long as we aim at a minimisation problem of a sum of powers of absolute values.

\vs
We shortly explain preceding studies of $p$-adic regression. Given an orthonormal Schauder basis, i.e.\ a topological basis compatible with the norm structure, of the Banach $\Qp$-vector space of $p$-adic continuous functions on the $D$-dimensional affine space $\Zp^D$ over the ring $\Zp$ of $p$-adic integers for a $D \in \N_{> 0}$, then truncation of the expression of a given continuous function as an infinite linear combination of the basis provides a approximation with respect to the $\ell^{\infty}$-norm. In particular, if we ignore noise or estimate the $\ell^{\infty}$-norm of noise as a small value, the approximation gives $p$-adic regression. K.\ Mahler introduced an orthonormal Schauder basis, which is nowadays called {\it Mahler basis}, of the Banach $\Qp$-vector space of $p$-adic continuous functions on $\Zp^D$ for a $D \in \N_{> 0}$ (cf.\ \cite{Mah58} Theorem 1 for $D = 1$ and \cite{Mah80} \S 12 for $D = 2$, which can be naturally generalised to higher dimensional cases). We extended Mahler basis to a $p$-adic character group of a formal group over $\Zp$ through the observation that Mahler's expansion is a $p$-adic analogue of Fourier transform (cf.\ \cite{Mih21} Theorem 3.23 and Theorem 3.24). Van der Put introduced another basis, which is nowadays called {\it van der Put basis}, for $D = 1$, and F.\ Bambozzi and we studied the notion of a {\it generalised van der Put basis} for an ultrametrisable topological space (cf.\ \cite{BM24} Proposition 7.7) and an extension of $p$-adic Stone--Weierstrass theorem (cf.\ \cite{Kap50} Theorem and \cite{Ber90} 9.2.5.\ Theorem for the original theorem and \cite{BM24} Theorem 7.10 for the extension).

\vs
S.\ Albeverio, A.\ Khrennikov, and B.\ Tirrozi formulated the prototype of $p$-adic neural network using the $p$-adic universal approximation theorem based on van der Put basis in \cite{AKT99} and \cite{KT00}, and G.\ L.\ R.\ N'guessan developed a new framework of $p$-adic neural network based on van der Put basis with experiment results by the WordNet data set in \cite{Ngu25}. A.\ P.\ Zubarev introduced a method to reduce the dimension by composing an explicit homeomorphism $\Zp^D \cong \Zp$ for the case $D > 1$ in \cite{Zub25-1}, and formulated a $D$-dimensional $p$-adic optimisation problem as a $1$-dimensional polynomial regression based on the $1$-dimensional Mahler basis in \cite{Zub25-2}.

\vs
E.\ Amaldi and V.\ Kann showed that the maximal feasible subsystem problem of linear equations over the finite field $\Fp$ is APX-complete, i.e.\ complete for the class of problems which allow constant-factor approximations, in \cite{AK95} Proposition A.1. Since the optimisation problem of modulo $p$ linear regression for the $\ell^1$-norm of the trivial valuation on $\Fp$ is identical to the maximal feasible subsystem problem of linear equations over $\Fp$ and is included in the optimisation problem of $p$-adic linear regression for $\ell^q$-norm of the $p$-adic valuation on $\Qp$ for $q \in (0,\infty)$, we need heuristic algorithms when we deal with $p$-adic linear regression.

\vs
Despite of the importance of $p$-adic regression in modern computer science, not so many are known on $p$-adic linear regression. G.\ D.\ Baker, S.\ Mccallum, and D.\ Pattinson gave a lower bound of the number of sample points in a hyperplane optimal of $p$-adic linear regression for $\ell^1$-norm in \cite{BMP25}. We formulated a new $p$-adic optimisation problem as an infinitesimal limit of the least squares method, invented and compared various heuristic algorithms for $p$-adic polynomial regression with thorough analysis of their time and space complexity in \cite{Mih26-1}.

\vs
In this paper, we introduce a new probabilistic algorithm (Algorithm \ref{trailig digits linear regression}) of $p$-adic linear regression. The method is based on the repetition of a new probabilistic algorithm (Algorithm \ref{linear regression mod p}) of modulo $p$ linear regression, which is based on the repetition of a new probabilistic algorithm (Algorithm \ref{noise-free locus}) essentially detecting inclusion of affine subspaces. The method works under a milder assumption on sampling than our preceding work on $p$-adic polynomial regression in \cite{Mih26-1}.

\vs
We briefly explain contents of this paper. In \S \ref{Repetitive Inclusion Decision}, we introduce the probabilistic algorithm (Algorithm \ref{noise-free locus}) essentially detecting inclusion of affine subspaces, the probabilistic algorithm (Algorithm \ref{linear regression mod p}) of modulo $p$ linear regression, and results on a specific probabilistic model (Theorem \ref{probability of dimension n - 1} and Theorem \ref{probability of inclusion}). Results of experiment appear there. In \S \ref{Digitwise Linear Regression}, we introduce the probabilistic algorithm (Algorithm \ref{trailing digits linear regression}) of $p$-adic linear regression.

\section{Convention}
\label{Convention}

We denote by $\N$ the set of non-negative integers. For sets $S$ and $T$, we denote by $T^S$ the set of maps $S \to T$. We note that every $d \in \N$ is identified with $\N \cap [0,d)$ in set theory, and hence $T^d$ for a set $T$ formally means $T^{\N \cap [0,d)}$, which is naturally identified with the set of $d$-tuples in $T$.

\vs
Throughout the paper, $p$ denotes a fixed prime number. For an $x \in \Zp$ and a power $q \in \N$ of $p$, we denote by $x \bmod q \in \Z/q \Z$ the reduction of $x$ modulo $q$. For a vector $\vec{x}$ of $p$-adic integers and a power $q \in \N$ of $p$, we denote by $\vec{x} \bmod q$ the vector of integers modulo $q$ given by the reduction modulo $q$ applied to each entry of $\vec{x}$.

\vs
When we write a pseudocode, a for-loop along a subset of $\N$ denotes the loop of the ascending order, and a for-loop along $I$ denotes a loop in an arbitrary order.

\section{Repetitive Inclusion Decision}
\label{Repetitive Inclusion Decision}

We introduce linear regression over the field $\Fp$ of integers modulo $p$ based on repetition of solving decision problems on inclusion.

\vs
Let $D$ be a non-negative integer. For a $\vec{c} = (c_d)_{d=0}^{D} \in \Fp^{D+1}$ and an $\vec{x} = (x_d)_{d=0}^{D-1} \in \Fp^D$, we set
\be
\left\langle \vec{c} , \vec{x} \right\rangle \coloneqq \sum_{d=0}^{D-1} c_d x_d + c_D.
\ee
When we refer to a {\it linear equation} in this section, we mean an equation on $(\vec{x},y) \in \Fp^D \times \Fp$ of the form
\be
y = \left\langle \vec{c} , \vec{x} \right\rangle,
\ee
and call $\vec{c}$ {\it the coefficient vector of the linear equation}.

\vs
Let $I$ be a finite non-empty set, $X = (\vec{x}_i)_{i \in I} \in (\Fp^D)^I$ a sequence of vectors of integers modulo $p$, $Y = (y_i)_{i \in I} \in \Fp^I$ a sequence of integers modulo $p$, $V \subset \Fp^D \times \Fp$ an affine subspace of codimension $1$, and $r \in [0,1)$. For a subset $W \subset \Fp^D \times \Fp$, we set $I_W \coloneqq \set{i \in I}{(\vec{x}_i,y_i) \in W}$.

\vs
We assume that the affine hull of $\set{(\vec{x}_i,y_i)}{i \in I_V}$ in $\Fp^D \times \Fp$ coincides with $V$, and the inequality
\be
\frac{\#(I \setminus I_V)}{\# I} \leq r \ll 2^{-1},
\ee
holds. We call $r$ {\it a noise probability bound}, and $I_V$ {\it the noise-free locus} respectively. We consider how to estimate $V$ from the data of $X$ and $Y$.

\subsection{Detecting Noise-Free locus}
\label{Detecting Noise-Free locus}

Let $I' \subset I$. We denote by $W$ the affine hull of $\set{(\vec{x}_i,y_i)}{i \in I'}$. We give a probabilistic algorithm to detect that $I'$ is a noise-free locus, i.e.\ $I' \subset I_V$.

\vs
If $I'$ is a noise-free locus, then $W \subset V$ holds by the monotonicity of affine hull. On the other hand, if $I'$ is not a noise-free locus, then at least one of $i \in I'$ satisfies $(\vec{x}_i,y_i) \notin V$, and hence $W \subset V$ does not hold. Therefore it suffices to give a probabilistic algorithm to detect the inclusion $W \subset V$.

\vs
Let $J$ be a finite set, and $C = (\vec{c}_j)_{j \in J} \in (\Fp^{D+1})^J$ a sequence of coefficient vectors. We assume that the corresponding system of linear equations defines $W$, i.e.\ 
\be
W = \set{(\vec{x},y) \in \Fp^D \times \Fp}{\forall j \in J, y = \langle \vec{c}_j , \vec{x} \rangle},
\ee
and is reduced in the sense that the codimension of $W$ as an affine subspace of $\Fp^D \times \Fp$ coincides with $\# J$. In particular, we have $\# J \leq D + 1$ and $\# I' \geq (D + 1) - \# J + 1 = D + 2 - \# J$.

\vs
We note that $C$ is a system of solutions $\vec{c}$ of the system
\be
\forall i \in I', y_i = \left\langle \vec{c} , \vec{x}_i \right\rangle
\ee
of linear equations, and hence can be explicitly constructed by Gauss elimination applied to $\set{(\vec{x}_i,1 \mid y_i)}{i \in I'}$.

\vs
In order to deduce whether the inclusion $W \subset V$ is likely to hold or not, we observe what likely occurs when $W \subset V$ and what likely occurs when $W \not\subset V$.

\vs
Suppose $\# I$ is sufficiently large, $X$ is sufficiently randomly chosen from $\Fp^D$, and $r$ is sufficiently smaller. Although the current assumptions are vague, we will give a formal argument later in \S \ref{Analysis of Probability}. If $W = V$, then we have
\be
\frac{\# I_W}{\# I} = \frac{\# I_V}{\# I} \geq 1 - r \approx 1.
\ee
If $W \subset V$, then we have an expectation
\be
\frac{\# I_W}{\# I} & \approx & (1 - r) \cdot \frac{\# W}{\# V} + r \cdot \frac{\# W}{\# \Fp^{D+1} - \# V} = (1 - r) \cdot \frac{p^{D + 1 - \#J}}{p^D} + r \cdot \frac{p^{D + 1 - \# J}}{p^{D + 1} - p^D} \\
& = & \left(1 - \frac{p - 2}{p - 1} r \right) p^{- \# J + 1} > \frac{9}{10} p^{- \# J + 1}.
\ee
If $W \not\subset V$, then we have an expectation
\be
\frac{\# I_W}{\# I} & \approx & (1 - r) \cdot \frac{\# (V \cap W)}{\# V} + r \cdot \frac{\# W - \#(V \cap W)}{\# \Fp^{D+1} - \# V} \leq (1 - r) \cdot \frac{p^{D - \# J}}{p^D} + r \cdot \frac{p^{D + 1 - \# J}}{p^{D + 1} - p^D} \\
& = & \left( 1 + \frac{r}{p - 1} \right) p^{- \# J} < \frac{9}{10} p^{- \#J + 1}.
\ee
Therefore, if $\# I_W / \# I$ is greater than the threshold $\frac{9}{10} p^{- \# J + 1}$, then it is likely that $W \subset V$. Here is a pseudocode of this criterion:

\begin{figure}[H]
\begin{algorithm}[H]
\caption{Dynamic variant of Gauss elimination applied to an extended row echelon form $A$ and the vector corresponding to $i$}
\label{dynamic Gauss elimination}
\begin{algorithmic}[1]
\Function {DynamicGaussElimination}{$p,X,Y,A,i$}
	\State $v = (v_d)_{d=0}^{D+1} \gets (\vec{x}_i,1 \mid y_i)$
	\State Apply row elementary transformations by rows of $A$ to $v$
	\State $\r{solvable} \gets$ True
	\If {$v$ is not a zero vector}
		\State $d \gets \min \set{d \in \N \cap [0,D+1]}{v_d \neq 0}$
		\If {$d = D + 1$}
			\State $\r{solvable} \gets$ False
		\Else
			\State $v \gets v_d^{-1} v$
			\State Apply row elementary transformation by $v$ to each row of $A$
			\State Insert $v$ to $A$
		\EndIf
	\EndIf
	\State \Return $(\r{solvable},A)$
\EndFunction
\end{algorithmic}
\end{algorithm}
\end{figure}

\begin{figure}[H]
\begin{algorithm}[H]
\caption{Detecting that $I'$ is a noise-free locus using its extended row echelon form $A$}
\label{noise-free matrix}
\begin{algorithmic}[1]
\Function {NoiseFreeMatrix}{$p,I,X,Y,I',A$}
	\State $C \gets$ the system of the coefficients vectors of the linear equations defining $W$ explicitly associated to $A$
	\State $c \gets 0$ \Comment{variable for $\# I_W$}
	\ForAll {$i \in I$}
		\State $\r{solution} \gets$ True \Comment{variable for whether the vector corresponding to $i$ is a solution of the system of linear equations corresponding to $C$}
		\ForAll {$\vec{c} \in C$}
			\If {$y_i \neq \langle \vec{c} , \vec{x}_i \rangle$}
				\State $\r{solution} \gets$ False
				\State \Break
			\EndIf
		\EndFor
		\If {$\r{solution}$}
			\State $c \gets c + 1$
		\EndIf
	\EndFor
	\State $L \gets$ the number of rows of $A$ \Comment{$\# J = (D + 1) - (L - 1) = D + 2 - L$}
	\State \Return $\frac{c - L}{\# I - L} > \frac{9}{10} p^{-(D + 1 - L)}$
\EndFunction
\end{algorithmic}
\end{algorithm}
\end{figure}

Optionally, it is good to separate $I$ into training data $I_0$ and validation data $I_1$, to restrict the process to the case $I' \subset I_0$, replacing $I$ in line 4 by $I_1$, and replacing $\frac{c - L}{\# I - L}$ by $\frac{c}{\# I_1}$ in order to remove the fitting/validation bias.

\begin{figure}[H]
\begin{algorithm}[H]
\caption{Detecting that $I'$ is a noise-free locus}
\label{noise-free locus}
\begin{algorithmic}[1]
\Function {NoiseFreeLocus}{$p,I,X,Y,I'$}
	\State $A \gets$ the empty matrix \Comment{variable for the extended row echelon form corresponding to $I'$}
	\ForAll {$i \in I'$}
		\State $(\r{solvable},A) \gets$ \Call{DynamicGaussElimination}{$p,X,Y,A,i$}
		\If {$\r{solvable}$ is False}
			\State \Return False
		\EndIf
	\EndFor
	\State \Return \Call{NoiseFreeMatrix}{$p,I,X,Y,I',A$}.
\EndFunction
\end{algorithmic}
\end{algorithm}
\end{figure}

The reason why we employ the dynamic variant (Algorithm \ref{dynamic Gauss elimination}) of Gauss elimination is because we will apply difference calculus to Gauss elimination in the next subsection. Concerning the implementation of the insertion of $v$ to $A$ in Algorithm \ref{dynamic Gauss elimination}, it suffices to implement $A$ as an variadic vector of vectors, to simply append $v$ to $A$, and to memorise the correspondence from each non-zero column to a row.

\vs
This criterion is weak against the data sparseness problem, and frequently fails when the assumptions that $\# I$ is sufficiently large, $X$ is sufficiently randomly chosen from $\Fp^D$, and $r$ is sufficiently small do not hold. Especially when $\# J$ is much greater than $\log_p \# I$, we have
\be
\frac{\# I_W}{\# I} \geq \frac{\# I'}{\# I} \geq \frac{1}{\# I} \gg \frac{9}{10} p^{- \#J + 1}
\ee
and hence this criterion always returns ``True'' in this case. Therefore, this criterion should be applied only when $\# J$ is not too large. In addition, if $p$ is $2$ or $3$, then the difference between the lower bound $(1 - \frac{p - 2}{p - 1} r) p^{- \# J + 1}$ for the case $W \subset V$ and the upper bound $( 1 + \frac{r}{p - 1}) p^{- \# J}$ for the case $W \not\subset V$ is relatively small, and hence a false-positive can occur more frequently.

\vs
In the next subsection, we apply this criterion under the caution above.

\subsection{Linear Regression Modulo $p$}
\label{Linear Regression Modulo p}

In order to estimate $V$, it suffices to construct a noise-free locus $I' \subset I$ with $\# I' = D + 1$ such that the affine hull $W$ of $I'$ is of codimension $1$ in $\Fp^D \times \Fp$, because then $W$ coincides with $V$.

\vs
For this purpose, it suffices to give a recursive way for a noise-free locus $I'$ with $\# I' < D + 1$ such that the affine hull of $I'$ is of codimension $D + 2 - \# I'$ in $\Fp^D \times \Fp$ to find an $i \in I \setminus I'$ such that $I' \cup \ens{i}$ is a noise-free locus and the affine hull of $I' \cup \ens{i}$ is of codimension $D + 1 - \# I'$ in $\Fp^D \times \Fp$.

\vs
In order to find such an $i \in I \setminus I'$, it suffices to randomly choose an $i \in I$ and check whether $I' \cup \ens{i}$ is a noise-free locus and the affine hull of $I' \cup \ens{i}$ is of codimension $D + 1 - \# I'$ in $\Fp^D \times \Fp$. In particular, Algorithm \ref{noise-free locus} solves the problem if $\# I'$ is not too small.

\vs
In order to solve the problem that Algorithm \ref{noise-free locus} does not work when $\# I'$ is too small, we combine RANSAC-like estimation and Algorithm \ref{noise-free locus}. When $\# I'$ is too small, we skip to check whether $I' \cup \ens{i}$ is a noise-free locus and the affine hull of $I' \cup \ens{i}$ is of codimension $D + 1 - \# I'$ in $\Fp^D \times \Fp$. Namely, we simply randomly extend $I'$ until $\# I'$ becomes larger than a suitable threshold $n \leq D + 1$. After then, it suffices to check whether $I'$ is still desired, i.e.\ $I'$ is a noise-free locus and the affine hull of $I'$ is of codimension $D + 2 - \# I'$ in $\Fp^D \times \Fp$. If $I'$ is not suitable, it is good to initialise $I'$ and restart the process. If $I'$ is still desired, then it is good to apply Algorithm \ref{noise-free locus}. Here is a pseudocode of the process:

\begin{figure}[H]
\begin{algorithm}[H]
\caption{Extending indices $I'$ up to the threshold $n$}
\label{extending indices 1}
\begin{algorithmic}[1]
\Function {ExtendingIndices1}{$p,I,X,Y,I',A,L,n$}
	\While {$L < n$}
		\State $i \gets$ a random element of $I$
		\State $(\r{solvable},A) \gets$ \Call{DynamicGaussElimination}{$p,X,Y,A,i$}
		\If {$\r{solvable}$}
			\If {$A$ has $L + 1$ rows} \Comment{detecting the linear independency of $i$ and $I'$, and in particular the non-duplication of the choice of $i$}
				\State Insert $i$ to $I'$
				\State $L \gets L + 1$
			\EndIf
		\Else
			\State \Break
		\EndIf
	\EndWhile
	\State \Return $(I',A,L)$
\EndFunction
\end{algorithmic}
\end{algorithm}
\end{figure}

\begin{figure}[H]
\begin{algorithm}[H]
\caption{Extending indices $I'$ beyond the threshold $n$}
\label{extending indices 2}
\begin{algorithmic}[1]
\Function {ExtendingIndices2}{$p,I,X,Y,I',A,L,\r{rep}$}
	\State $c \gets 0$ \Comment{variable for the number of trials}
	\While {$L < D + 1$ and $c = 0$}
		\While {$c < \r{rep}$}
			\State $i \gets$ a random element of $I$
			\State $(\r{solvable},B) \gets$ \Call{DynamicGaussElimination}{$p,X,Y,A,i$} \Comment{the input $A$ is copied and is not changed}
			\If {$\r{solvable}$ is False, the number of rows of $B$ is $L$, or \Call{NoiseFreeMatrix}{$p,I,X,Y,I' \cup \ens{i},B$} is False}
				\State $c \gets c + 1$
				\State \Continue
			\EndIf
			\State $c \gets 0$
			\State $A \gets B$
			\State $L \gets L + 1$
			\State Insert $i$ to $I'$
			\State \Break
		\EndWhile
	\EndWhile
	\State \Return $(I',A,L)$
\EndFunction
\end{algorithmic}
\end{algorithm}
\end{figure}

If $I$ is separated into training data $I_0$ and validation data $I_1$ as the explanation after Algorithm \ref{noise-free matrix}, $I$ in line 3 should be replaced by $I_0$.

\begin{figure}[H]
\begin{algorithm}[H]
\caption{Estimation of $V$ from $(I,X,Y)$}
\label{linear regression mod p}
\begin{algorithmic}[1]
\Function {LinearRegressionModulo}{$p,I,X,Y,\r{rep}$}
	\State $n \gets \max(1, \min(\# I - 1, D + 1 - \lfloor \log_p \# I \rfloor))$ \Comment{threshold of $\# I'$}
	\While {True}
		\State $I' \gets \emptyset$
		\State $A \gets$ an empty matrix \Comment{variable for the extended row echelon form corresponding to $I'$}
		\State $L \gets 0$ \Comment{variable for the number of rows of $A$}
		\State $(I',A,L) \gets$ \Call{ExtendingIndices1}{$p,I,X,Y,I',A,L,n$}
		\If {$L < n$ or \Call{NoiseFreeMatrix}{$p,I,X,Y,I',A$} is False}
			\State \Continue
		\EndIf
		\State $(I',A,L) \gets$ \Call{ExtendingIndices2}{$p,I,X,Y,I',A,L,\r{rep}$}
		\If {$L = D + 1$}
			\State \Return the coefficient vector $\vec{c}$ of the linear equation defining $W$ explicitly associated to $A$
		\EndIf
	\EndWhile
\EndFunction
\end{algorithmic}
\end{algorithm}
\end{figure}

If $I$ is separated into training data $I_0$ and validation data $I_1$ as the explanation after Algorithm \ref{noise-free matrix}, $I$ in line 5 should be replaced by $I_0$.

\vs
The hyper parameter $\r{rep}$ of the number of retrials to search a new $i$ to extend $I'$ should not be much larger than the expected value $(1 - (1 - p^{-1})r)^{-1}$ of the number of trials when $I'$ is actually a noise-free locus. Under the assumption $r \ll 2^{-1}$, it suffices to set $\r{rep}$ as $2$ or $3$. Concerning an implementation of $I'$, it suffices to simply implement $I'$ as a variadic vector, because we do not use the order or the membership relation.

\subsection{Analysis of Probability}
\label{Analysis of Probability}

We formalise the intuitive assumptions in \S \ref{Detecting Noise-Free locus} and \S \ref{Linear Regression Modulo p}. Let $D$, $I$, $r$, and $V$ be the same as in the beginning of \S \ref{Repetitive Inclusion Decision}.

\vs
Let $(\Omega,\cP)$ be a probability space. A {\it random sampling of $V$ in $\Fp^D \times \Fp$ of noise probability $r$} is a tuple $(\vec{x},y,b)$ of a uniform random variable $\vec{x} \colon \Omega \to \Fp^D$, a random variable $y \colon \Omega \to \Fp$, and a random variable $b \colon \Omega \to \ens{0,1}$ with $P(\set{\omega \in \Omega}{b(\omega) = 0}) = r$ satisfying the following:
\bi
\item[(1)] If $r > 0$, then $\vec{x}$ and $y$ restricted to $\set{\omega \in \Omega}{b(\omega) = 0}$ are uniform.
\item[(2)] The probability of $(\vec{x}(\omega),y(\omega)) \in V$ restricted to $\set{\omega \in \Omega}{b(\omega) = 1}$ is $1$.
\ei
Let $((\vec{x}_i,\vec{y}_i,b_i))_{i \in I}$ be a sequence of random samplings of $V$ in $\Fp^D \times \Fp$ of noise probability $r$. Set $X \coloneqq (\vec{x}_i)_{i \in I}$ and $Y \coloneqq (y_i)_{i \in I}$. We assume that entries of $X$ are independent, entries of $Y$ are independent, and entries of $(b_i)_{i \in I}$ are independent.

\begin{thm}
\label{probability of dimension n - 1}
Set $u \coloneqq 1 - (1 - p^{-1}) r$. For an $\omega \in \Omega$, let $I'(\omega) \subset I$ denote $I'$ at the first arrival of line 8 in Algorithm \ref{linear regression mod p} for \Call{LinearRegressionModulo}{$p,I,X(\omega),Y(\omega),\r{rep}$} for a $\r{rep} \in \N$, and $W(\omega) \subset \Fp^D \times \Fp$ denote the affine hull of $\set{(\vec{x}_i(\omega),y_i(\omega))}{i \in I'(\omega)}$. Then the probability $P$ that $W(\omega)$ is an affine subset of $V$ of dimension $n - 1$ satisfies
\be
u^n \prod_{d=0}^{n-2} (1 - p^{-(D - d)}) \leq P \leq \frac{u^n \prod_{d=0}^{n-2} (1 - p^{-(D - d)})}{(1 - p^{-(D + 2 - n)} u)^{n - 1}},
\ee
where $n$ denotes the threshold in line 2 in Algorithm \ref{linear regression mod p}. 
\end{thm}

\begin{proof}
We have $0 < p^{-(D + 2 - n)} u < 1$. In particular, the right hand side in the assertion makes sense. For any affine subset $W' \subset V$ of dimension $d \in \N \cap [0,D]$ and any $i \in I$,
we have
\be
\cP(\set{\omega \in \Omega}{(\vec{x}_i(\omega),y_i(\omega)) \in W'}) & = & (1 - r) \times \frac{\# W'}{\# V} + r \times \frac{\# W'}{\#(\Fp^D \times \Fp)} \\
& = & (1 - r) \times p^{-(D - d)} + r \times p^{-(D - d + 1)} \\
& = & p^{-(D - d)} u,
\ee
and hence
\be
\cP(\set{\omega \in \Omega}{(\vec{x}_i(\omega),y_i(\omega)) \in V \setminus W'}) & = & p^{-(D - D)} u - p^{-(D - d)} u \\
& = & (1 - p^{-(D - d)}) u.
\ee
The assertion for the case $n = 1$ follows from the first formula applied to $W' = V$. Suppose $n > 1$. Let $\vec{i} = (i_h)_{h=0}^{\# I - 1}$ be a permutation of $I$, which plays a role of the time series data of the non-duplicated values $i$ in line 3 in Algorithm \ref{extending indices 1}. For any $H \in \N \cap [0,\# I]$, we have
\be
\cP(\set{\omega \in \Omega}{\forall h \in \N \cap [0,H)[(\vec{x}_{i_h}(\omega),y_{i_h}(\omega)) \in V]}) = u^H.
\ee
For each $H \in \N$, we set $S_H \coloneqq \set{(h_d)_{d=0}^{n-2} \in \N^{n-1}}{\sum_{d=0}^{n-2} h_d = H - n}$. For each $(H,k) \in \N^2$, we set $C_{H,k} \coloneqq \# \set{(h_d)_{d=0}^{n-2} \in S_H}{\sum_{d=0}^{n-2} d h_d = k}$. For any $H \in \N \cap [n,\# I]$, the probability $P_H$ that the affine hull of $\set{(\vec{x}_{i_h},y_{i_h})}{h \in \N \cap [0,H-1)}$ is not an affine subset of $V$ of dimension $n - 1$ and the affine hull of $\set{(\vec{x}_{i_h},y_{i_h})}{h \in \N \cap [0,H)}$ is an affine subset of $V$ of dimension $n - 1$, which does not depend on $\vec{i}$, satisfies
\be
P_H & = & \sum_{(h_d)_{d=0}^{n-2} \in S_H} u \prod_{d=0}^{n-2} \left( p^{-(D - d)} u \right)^{h_d} (1 - p^{-(D - d)}) u \\
& = & u^H \sum_{(h_d)_{d=0}^{n-2} \in S_H} \prod_{d=0}^{n-2} p^{-(D - d)h_d}(1 - p^{-(D - d)}) \\
& = & u^H p^{-D(H - n)} \left( \prod_{d=0}^{n-2} (1 - p^{-(D - d)}) \right) \sum_{(h_d)_{d=0}^{n-2} \in S_H} \prod_{d=0}^{n-2} p^{d h_d} \\
& = & u^H p^{-D(H - n)} \left( \prod_{d=0}^{n-2} (1 - p^{-(D - d)}) \right) \sum_{(h_d)_{d=0}^{n-2} \in S_H} p^{\sum_{d=0}^{n-2} d h_d} \\
& = & u^H p^{-D(H - n)} \left( \prod_{d=0}^{n-2} (1 - p^{-(D - d)}) \right) \sum_{k=0}^{(n-2)(H-n)} C_{H,k} p^k.
\ee
We have $P = \sum_{H=n}^{\# I} P_H$. By $n \leq \# I$, we obtain
\be
P \geq P_n = u^n p^0 \left( \prod_{d=0}^{n-2} (1 - p^{-(D - d)}) \right) C_{n,0} p^0 = u^n \prod_{d=0}^{n-2} (1 - p^{-(D - d)}).
\ee
On the other hand, we have
\be
P & = & \sum_{H=n}^{\# I} P_H \\
& = & \sum_{H=n}^{\# I} u^H p^{-D(H - n)} \left( \prod_{d=0}^{n-2} (1 - p^{-(D - d)}) \right) \sum_{k=0}^{(n-2)(H-n)} C_{H,k} p^k \\
& \leq & \sum_{H=n}^{\# I} u^H p^{-(D + 2 - n)(H - n)} \left( \prod_{d=0}^{n-2} (1 - p^{-(D - d)}) \right) \sum_{k=0}^{(n-2)(H-n)} C_{H,k} \\
& = & \sum_{H=n}^{\# I} u^H p^{-(D + 2 - n)(H - n)} \left( \prod_{d=0}^{n-2} (1 - p^{-(D - d)}) \right) \# S_H \\
& = & \sum_{H=n}^{\# I} u^H p^{-(D + 2 - n)(H - n)} \left( \prod_{d=0}^{n-2} (1 - p^{-(D - d)}) \right) \binom{H - 2}{n - 2} \\
& = & \sum_{H'=0}^{\# I - n} u^{H' + n} p^{-(D + 2 - n)H'} \left( \prod_{d=0}^{n-2} (1 - p^{-(D - d)}) \right) \binom{H' + n - 2}{n - 2} \\
& = & u^n \left( \prod_{d=0}^{n-2} (1 - p^{-(D - d)}) \right) \sum_{H'=0}^{\# I - n} \binom{H' + n - 2}{H'} \left( p^{-(D + 2 - n)} u \right)^{H'} \\
& \leq & u^n \left( \prod_{d=0}^{n-2} (1 - p^{-(D - d)}) \right) \sum_{H'=0}^{\infty} \binom{H' + n - 2}{H'} \left( p^{-(D + 2 - n)} u \right)^{H'} \\
& = & u^n \left( \prod_{d=0}^{n-2} (1 - p^{-(D - d)}) \right) \sum_{H'=0}^{\infty} \binom{- n + 1}{H'} \left( -p^{-(D + 2 - n)} u \right)^{H'} \\
& = & u^n \left( \prod_{d=0}^{n-2} (1 - p^{-(D - d)}) \right) \left( 1 - p^{-(D + 2 - n)} u \right)^{- n + 1} \\
& = & \frac{u^n \prod_{d=0}^{n-2} (1 - p^{-(D - d)})}{(1 - p^{-(D + 2 - n)} u)^{n - 1}}.
\ee
\end{proof}

Although Theorem \ref{probability of dimension n - 1} does not give estimation for the second or later arrival, it is natural to expect how many times the process arrives line 2 in Algorithm \ref{linear regression mod p} until the first accomplishment that the affine hull of the rows of $A$ is an affine subset of $V$ of dimension $n - 1$ is approximately bounded above by
\be
P^{-1} \leq u^{-n} \prod_{d=0}^{n-2} (1 - p^{-(D - d)})^{-1}.
\ee
We next formalise the argument on $\# I_W/\# I$ in \S \ref{Detecting Noise-Free locus}.

\begin{thm}
\label{probability of inclusion}
Let $W \subset \Fp^D \times \Fp$ be an affine subset. For any $K \in \N \cap [0,\# I]$ and $i \in I$, the probability $P_{W,\leq K}$ of $\# \set{i \in I}{(\vec{x}_i(\omega),y_i(\omega)) \in W} \leq K$ satisfies
\be
P_{W,\leq K} = F \left( K ; \# I , p^{- D}(\#(V \cap W) - (\#(V \cap W) - p^{- 1} \#W)r) \right).
\ee
where $F(x;N,q)$ for an $(N,q) \in \N \times [0,1]$ denotes the cumulative distribution function for the binomial distribution for $N$ independent experiments with success probability $q$.
\end{thm}

\begin{proof}
For each $i \in I$, we denote by $\delta_i$ the random variable
\be
\Omega & \to & \ens{0,1} \\
\omega & \mapsto & 
\left\{
\begin{array}{ll}
1 & ((\vec{x}_i(\omega),y_i(\omega)) \in W) \\
0 & ((\vec{x}_i(\omega),y_i(\omega)) \notin W)
\end{array}
\right..
\ee
Then we have
\be
\# \set{i \in I}{(\vec{x}_i(\omega),y_i(\omega)) \in W} = \sum_{i \in I} \delta_i(\omega).
\ee
For each $k \in \N \cap [0,\# I]$, we set
\be
P_{W,k} \coloneqq \cP \left( \set{\omega \in \Omega}{\# \set{i \in I}{(\vec{x}_i(\omega),y_i(\omega)) \in W} = k} \right).
\ee
Set $C \coloneqq \# W$ and $C' \coloneqq \#(V \cap W)$. For any $i \in I$, we have
\be
\cP(\set{\omega \in \Omega}{\delta_i(\omega) = 1}) & = & (1 - r) \times \frac{C'}{\# V} + r \times \frac{C}{\#(\Fp^D \times \Fp)} \\
& = & (1 - r) \times \frac{C'}{p^D} + r \times \frac{C}{p^{D+1}} = p^{- D}(C' - (C' - p^{-1}C)r).
\ee
Therefore, for any $k \in \N \cap [0,\# I]$, we have
\be
P_{W,k} & = & \binom{\# I}{k} (p^{- D}(C' - (C' - p^{-1}C)r))^k (1 - p^{- D}(C' - (C' - p^{-1}C)r))^{\# I - k}.
\ee
This implies
\be
P_{W,\leq K} & = & \sum_{k=0}^{K} P_{W,k} \\
& = & \sum_{k=0}^{K} \binom{\# I}{k} (p^{- D}(C' - (C' - p^{-1}C)r))^k (1 - p^{- D}(C' - (C' - p^{-1}C)r))^{\# I - k} \\
& = & F \left( K ; \# I , p^{- D}(C' - (C' - p^{-1}C)r) \right).
\ee
\end{proof}

Let $W \subset \Fp^D \times \Fp$ be an affine subset defined by a reduced system of linear equations indexed by a finite set $J$, as in \S \ref{Detecting Noise-Free locus}. Then Theorem \ref{probability of inclusion} implies
\be
P_{W,\leq \lfloor \frac{9}{10} p^{-D} \# I \rfloor} \ 
\left\{
\begin{array}{cll}
= & F \left( \left\lfloor \frac{9}{10} p^{- \# J + 1} \# I \right\rfloor ; \# I , p^{- \# J + 1} u \right) & (W \subset V) \\
\geq & F \left( \left\lfloor \frac{9}{10} p^{- \# J + 1} \# I \right\rfloor ; \# I , p^{- \# J} \right) & (W \not\subset V)
\end{array}
\right..
\ee
The right hand side for the case $W \subset V$ converges to $0$ as $\# I \to \infty$ if $r < (10(1 - p^{-1}))^{-1}$, and that for the case $W \not\subset V$ converges to $1$ as $\# I \to \infty$ by law of large number. This means that Algorithm \ref{noise-free matrix} detects the noise-freeness of $I'$ with low false-negative rate and low false-positive rate if $I$ is sufficiently large.

\subsection{Experiment}
\label{Experiment}

We applied Algorithm \ref{linear regression mod p} to random cases with $p = 7$, $I = \N \cap [0,10^5)$, and $\r{rep} = 3$. In the experiment, we chose a random non-zero vector $\vec{c} \in \Fp^{D+1}$ to define $V \coloneqq \set{(\vec{x},y) \in \Fp^D \times \Fp}{y = \langle \vec{c} , \vec{x} \rangle}$, random vectors $X = (\vec{x}_i)_{i \in I} \in (\Fp^D)^I$, a random subset $I' \subset I$ with $\# (I \setminus I') \approx r \# I$, a vector $Y = (y_i)_{i \in I}$ by setting $y_i \coloneqq \langle \vec{c} , \vec{x}_i \rangle$ for each $i \in I'$ and randomly choosing $y_i \in \Fp$ for each $i \in I \setminus I'$ literally by the python code
\begin{lstlisting}
import random
random.seed(100+case)
def subs(c,x):return sum(c[d]*x[d] for d in range(D+1))%p
c=[random.randint(0,p-1) for d in range(D+1)]
X=[[random.randint(0,p-1) for d in range(D)] for i in range(N)]
Y=[random.randint(0,p-1) if random.randint(0,99) < R else subs(c,x+[1]) for x in X]
\end{lstlisting}
with $N \coloneqq \# I = 10^5$, and $R \coloneqq 100r$, where ``case'' denotes the case number of each process. In order to express the performance in a way independent of machine specification, we list
\bi
\item[$t$:] the case number,
\item[$c_0$:] how many times the initialisation $I' = \emptyset$ is retried after failure,
\item[$c_1$:] how many times the searching of a new $i$ to extend $I'$ is retried after failure when the threshold condition $\# I' \geq n$ holds, and
\item[$s$:] whether the linear regression returns a correct coefficient vector $\vec{c}$ or not (T or F for short).
\ei

\begin{table}[H]
\begin{center}
\begin{minipage}{0.45\textwidth}
\centering
\caption{$D = 20, r = 0.01$}
\begin{tabular}{|r||r|r|r|}
\hline
$t$ & $c_0$ & $c_1$ & $s$ \\
\hline
\hline
0 & 0 & 1 & T \\
\hline
1 & 0 & 0 & T \\
\hline
2 & 0 & 0 & T \\
\hline
3 & 0 & 0 & T \\
\hline
4 & 0 & 0 & T \\
\hline
5 & 0 & 0 & T \\
\hline
6 & 0 & 1 & T \\
\hline
7 & 0 & 0 & T \\
\hline
8 & 0 & 1 & T \\
\hline
9 & 0 & 0 & T \\
\hline
\end{tabular}
\end{minipage}
\begin{minipage}{0.45\textwidth}
\centering
\caption{$D = 20, r = 0.03$}
\begin{tabular}{|r||r|r|r|}
\hline
$t$ & $c_0$ & $c_1$ & $s$ \\
\hline
\hline
0 & 1 & 3 & T \\
\hline
1 & 0 & 0 & T \\
\hline
2 & 1 & 3 & T \\
\hline
3 & 0 & 0 & T \\
\hline
4 & 0 & 1 & T \\
\hline
5 & 0 & 0 & T \\
\hline
6 & 0 & 0 & T \\
\hline
7 & 0 & 0 & T \\
\hline
8 & 0 & 1 & T \\
\hline
9 & 0 & 0 & T \\
\hline
\end{tabular}
\end{minipage}

\vs
\begin{minipage}{0.45\textwidth}
\centering
\caption{$D = 40, r = 0.01$}
\begin{tabular}{|r||r|r|r|}
\hline
$t$ & $c_0$ & $c_1$ & $s$ \\
\hline
\hline
0 & 0 & 0 & T \\
\hline
1 & 1 & 3 & T \\
\hline
2 & 0 & 0 & T \\
\hline
3 & 0 & 1 & T \\
\hline
4 & 0 & 2 & T \\
\hline
5 & 0 & 0 & T \\
\hline
6 & 0 & 2 & T \\
\hline
7 & 1 & 3 & T \\
\hline
8 & 0 & 0 & T \\
\hline
9 & 0 & 0 & T \\
\hline
\end{tabular}
\end{minipage}
\begin{minipage}{0.45\textwidth}
\centering
\caption{$D = 40, r = 0.03$}
\begin{tabular}{|r||r|r|r|}
\hline
$t$ & $c_0$ & $c_1$ & $s$ \\
\hline
\hline
0 & 0 & 0 & T \\
\hline
1 & 2 & 7 & T \\
\hline
2 & 0 & 0 & T \\
\hline
3 & 0 & 0 & T \\
\hline
4 & 2 & 8 & T \\
\hline
5 & 0 & 0 & T \\
\hline
6 & 2 & 8 & T \\
\hline
7 & 0 & 0 & T \\
\hline
8 & 0 & 0 & T \\
\hline
9 & 0 & 0 & T \\
\hline
\end{tabular}
\end{minipage}
\end{center}
\end{table}

\begin{table}[H]
\begin{center}
\begin{minipage}{0.45\textwidth}
\centering
\caption{$D = 60, r = 0.01$}
\begin{tabular}{|r||r|r|r|}
\hline
$t$ & $c_0$ & $c_1$ & $s$ \\
\hline
\hline
0 & 0 & 0 & T \\
\hline
1 & 0 & 1 & T \\
\hline
2 & 0 & 0 & T \\
\hline
3 & 0 & 0 & T \\
\hline
4 & 0 & 1 & T \\
\hline
5 & 1 & 3 & T \\
\hline
6 & 2 & 6 & T \\
\hline
7 & 1 & 3 & T \\
\hline
8 & 0 & 0 & T \\
\hline
9 & 0 & 0 & T \\
\hline
\end{tabular}
\end{minipage}
\begin{minipage}{0.45\textwidth}
\centering
\caption{$D = 60, r = 0.03$}
\begin{tabular}{|r||r|r|r|}
\hline
$t$ & $c_0$ & $c_1$ & $s$ \\
\hline
\hline
0 & 3 & 10 & T \\
\hline
1 & 0 & 0 & T \\
\hline
2 & 0 & 0 & T \\
\hline
3 & 0 & 2 & T \\
\hline
4 & 2 & 6 & T \\
\hline
5 & 2 & 6 & T \\
\hline
6 & 0 & 0 & T \\
\hline
7 & 1 & 3 & T \\
\hline
8 & 1 & 3 & T \\
\hline
9 & 0 & 1 & T \\
\hline
\end{tabular}
\end{minipage}

\vs
\begin{minipage}{0.45\textwidth}
\centering
\caption{$D = 80, r = 0.01$}
\begin{tabular}{|r||r|r|r|}
\hline
$t$ & $c_0$ & $c_1$ & $s$ \\
\hline
\hline
0 & 2 & 6 & T \\
\hline
1 & 1 & 4 & T \\
\hline
2 & 0 & 0 & T \\
\hline
3 & 1 & 3 & T \\
\hline
4 & 0 & 0 & T \\
\hline
5 & 3 & 9 & T \\
\hline
6 & 1 & 3 & T \\
\hline
7 & 0 & 0 & T \\
\hline
8 & 3 & 9 & T \\
\hline
9 & 3 & 9 & T \\
\hline
\end{tabular}
\end{minipage}
\begin{minipage}{0.45\textwidth}
\centering
\caption{$D = 80, r = 0.03$}
\begin{tabular}{|r||r|r|r|}
\hline
$t$ & $c_0$ & $c_1$ & $s$ \\
\hline
\hline
0 & 2 & 6 & T \\
\hline
1 & 0 & 0 & T \\
\hline
2 & 13 & 40 & T \\
\hline
3 & 7 & 21 & T \\
\hline
4 & 5 & 15 & T \\
\hline
5 & 11 & 34 & T \\
\hline
6 & 9 & 27 & T \\
\hline
7 & 12 & 36 & T \\
\hline
8 & 3 & 9 & T \\
\hline
9 & 5 & 16 & T \\
\hline
\end{tabular}
\end{minipage}

\vs
\begin{minipage}{0.45\textwidth}
\centering
\caption{$D = 100, r = 0.01$}
\begin{tabular}{|r||r|r|r|}
\hline
$t$ & $c_0$ & $c_1$ & $s$ \\
\hline
\hline
0 & 1 & 3 & T \\
\hline
1 & 2 & 6 & T \\
\hline
2 & 0 & 0 & T \\
\hline
3 & 4 & 12 & T \\
\hline
4 & 0 & 0 & T \\
\hline
5 & 1 & 3 & T \\
\hline
6 & 7 & 21 & T \\
\hline
7 & 0 & 0 & T \\
\hline
8 & 0 & 0 & T \\
\hline
9 & 0 & 0 & T \\
\hline
\end{tabular}
\end{minipage}
\begin{minipage}{0.45\textwidth}
\centering
\caption{$D = 100, r = 0.03$}
\begin{tabular}{|r||r|r|r|}
\hline
$t$ & $c_0$ & $c_1$ & $s$ \\
\hline
\hline
0 & 14 & 43 & T \\
\hline
1 & 2 & 6 & T \\
\hline
2 & 29 & 87 & T \\
\hline
3 & 14 & 43 & T \\
\hline
4 & 1 & 3 & T \\
\hline
5 & 5 & 15 & T \\
\hline
6 & 0 & 2 & T \\
\hline
7 & 9 & 27 & T \\
\hline
8 & 37 & 111 & T \\
\hline
9 & 2 & 7 & T \\
\hline
\end{tabular}
\end{minipage}
\end{center}
\end{table}

We note that the process does not necessarily terminate, but at least all the processes in the experiment above are marked as T, i.e. terminated with correct returns $\vec{c}$.

\vs
By Theorem \ref{probability of dimension n - 1}, the value
\be
\frac{(1 - p^{-(D + 2 - n)} u)^{n - 1}}{u^n \prod_{d=0}^{n-2} (1 - p^{-(D - d)})} = u^{-1} \prod_{d=0}^{n-2} \frac{u^{-1} - p^{-(D + 2 - n)}}{1 - p^{-(D - d)}}
\ee
is a lower bound of the expected value of $c_0$, and becomes large when $r$ and $D$ are not so small. In particular, when $D = 100$ and $r = 0.1$, then $c_0$ is expected to be larger than $4979$, and indeed the process did not terminate in $1700$ trials of initialisation $I' = \emptyset$.

\section{Digitwise Linear Regression}
\label{Digitwise Linear Regression}

We introduce linear regression over the ring $\Zp$ of $p$-adic integers based on repetition of linear regression modulo $p$. We note that multiplying by a power of $p$, we can reduce linear regression over $\Qp$ to that over $\Zp$.

\vs
Let $D$ be a non-negative integer. For a $\vec{c} = (c_d)_{d=0}^{D} \in \Zp^{D+1}$ and an $\vec{x} = (x_d)_{d=0}^{D-1} \in \Zp^D$, we set
\be
\left\langle \vec{c} , \vec{x} \right\rangle \coloneqq \sum_{d=0}^{D-1} c_d x_d + c_D.
\ee
by abusing the notation for $\Fp$ in \S \ref{Repetitive Inclusion Decision}. When we refer to a {\it linear equation} in this section, we mean an equation on $(\vec{x},y) \in \Zp^D \times \Zp$ of the form
\be
y = \left\langle \vec{c} , \vec{x} \right\rangle,
\ee
and call $\vec{c}$ {\it the coefficient vector of the linear equation}.

\vs
A subset $W \subset \Zp^D \times \Zp$ is said to be an {\it affine subspace of codimension $1$} if there exists a $\vec{c} \in \Zp^{D+1}$ such that $W = \set{(\vec{x},y) \in \Zp^D \times \Zp}{y = \langle \vec{c} , \vec{x} \rangle}$. We call such a $\vec{c}$ a {\it defining coefficient vector} of $W$. For a non-empty subset $S \subset \Zp^D \times \Zp$, we define its {\it affine hull} in $\Zp^D \times \Zp$ as the intersection of all affine subspace of codimension $1$ containing $S$, which is formally defined to be the ambient space $\Zp^D \times \Zp$ when there is no such affine subspace.

\vs
Let $I$ be a finite non-empty set, $X = (\vec{x}_i)_{i \in I} \in (\Zp^D)^I$ a sequence of vectors of $p$-adic integers, $Y = (y_i)_{i \in I} \in \Zp^I$ a sequence of $p$-adic integers, $V \subset \Zp^D \times \Zp$ an affine subspace of codimension $1$, $\vec{c} \in \Zp^{D+1}$ a defining coefficient of $V$, and $r \in [0,1)$. We set
\be
I_e \coloneqq \set{i \in I}{y_i - \left\langle \vec{c} , \vec{x}_i \right\rangle \in p^e \Zp}
\ee
for each $e \in \N$.

\vs
Let $e \in \N$. For any $(\vec{x},y) \in V$, we have
\be
\frac{y - \left\langle \vec{c} \bmod p^e , \vec{x} \right\rangle}{p^e} = \left\langle \frac{\vec{c} - (\vec{c} \bmod p^e)}{p^e} , \vec{x} \right\rangle,
\ee
where $\vec{c} \bmod p^e \in (\Z/p^e \Z)^{D+1}$ is naturally identified with its representative in $(\N \cap [0,p^e))^{D+1}$. This implies that $\set{(\vec{x}, p^{-e}(y - \langle \vec{c} \bmod p^e , \vec{x} \rangle))}{(\vec{x},y) \in V}$ is an affine subvariety of $\Zp^D \times \Zp$ of codimension $\leq 1$ with defining coefficient $p^{-e}(\vec{c} - (\vec{c} \bmod p^e))$.

\vs
Let $E \in \N \cap [1,\infty)$. We assume that the affine hull of $\set{(\vec{x}_i \bmod p, y_i \bmod p)}{i \in I_E}$ in $\Fp^D \times \Fp$ coincides with $\ol{V} \coloneqq \set{\vec{v} \bmod p}{\vec{v} \in V} \subset \Fp^D \times \Fp$, and the inequality
\be
\frac{\#(I_e \setminus I_{e+1})}{\# I_e} \leq r,
\ee
holds for any $e \in \N \cap [0,E)$. We consider how to estimate $\vec{c}$ modulo $p^E$ from the data of $X$ and $Y$.

\subsection{Estimation of the Last Digit}
\label{Estimation of the Last Digit}

We set $\ol{X} \coloneqq (\vec{x}_i \bmod p)_{i \in I} \in \Fp^D$ and $\ol{Y} \coloneqq Y \bmod p$. We note that if $X$ is concentrated on a $p$-adic neighbourhood of an affine subspace of $\Zp^D$ of positive codimension, $\ol{X}$ can be a subset of an affine subspace of $\Fp^D$ of positive codimension.

\vs
We assume that $\ol{X}$ is sufficiently random. Formally speaking, we assume that the affine hull of $\set{(\vec{x}_i \bmod p , y_i \bmod p)}{i \in I_1}$ in $\Fp^D \times \Fp$ coincides with $\ol{V}$ and $((\vec{x}_i \bmod p , y_i \bmod p))_{i \in I_1}$ is derived from a sequence of random samplings of $\ol{V}$ in $\Fp^D \times \Fp$ of noise probability $r$ in the sense of \S \ref{Analysis of Probability}. By $I = I_0$, we have
\be
\frac{\#(I \setminus I_1)}{\# I} = \frac{\#(I_0 \setminus I_1)}{\# I_0} \leq r,
\ee
and $I_1 = \set{i \in I}{(\vec{x} \bmod p \rangle,y_i \bmod p) \in \ol{V}}$. Therefore, $(I,\ol{X},\ol{Y})$ satisfies the assumptions of \S \ref{Repetitive Inclusion Decision}.

\vs
Applying Algorithm \ref{linear regression mod p} to $(I,\ol{X},\ol{Y})$, we obtain a vector $\theta \in \Fp^{D+1}$ of integers modulo $p$ such that the equality $\ol{V} = \set{(\vec{x},y) \in \Fp^D \times \Fp}{y = \langle \theta , \vec{x}\rangle}$ is expected to hold. By the uniqueness of a coefficient vector defining a horizontal affine subspace of $\Fp^D \times \Fp$ of codimension $1$, we have $\theta = \vec{c} \bmod p$ if the expectation is actually true. Therefore, $\theta$ is an estimation of $\vec{c} \bmod p$, i.e.\ the vector of the last digits of entries of $\vec{c}$. Here is a pseudocode of this process:

\begin{figure}[H]
\begin{algorithm}[H]
\caption{Estimation of $\vec{c} \bmod p$ from $(I,X,Y)$}
\label{last digit linear regression}
\begin{algorithmic}[1]
\Function {LastDigitLinearRegression}{$p,I,X,Y,\textrm{rep}$}
	\State $\theta \gets$ \Call{LinearRegressionModulo}{$p,I,\ol{X},\ol{Y},\textrm{rep}$}
	\State \Return the vector of the unique representatives in $\N \cap [0,p)$ of entries of $\theta$
\EndFunction
\end{algorithmic}
\end{algorithm}
\end{figure}

\subsection{Estimation of the Trailing Digits}
\label{Estimation of the Trailing Digits}

Assuming that $\ol{X}$ is sufficiently random and the expectation $\theta = \vec{c} \bmod p$ holds, we next explain how to estimate the second last digits of entries of $\vec{c}$. We denote by $\tl{\theta} \in (\N \cap [0,p))^{D+1}$ the return value of \Call{LastDigitLinearRegression}{$(p,I,X,Y,\textrm{rep})$}, where $\textrm{rep}$ is the hyper parameter used in this context. We set
\be
\vec{c}_1 & \coloneqq & p^{-1} \left( \vec{c} - \tl{\theta} \right) \in \Qp^{D+1} \\
Y_1 = (y_{1,i})_{i \in I_1} & \coloneqq & p^{-1} \left( y_i - \left\langle \tl{\theta} , \vec{x}_i \right\rangle \right)_{i \in I_1} \in \Qp^{I_1}
\ee
By the assumption, we have $\vec{c}_1 \in \Zp^{D+1}$ and $Y_1 \in \Zp^{I_1}$. For any $i \in I_1$, we have
\be
y_{i,1} - \left\langle \vec{c}_1 , \vec{x}_i \right\rangle & = & p^{-1} \left( y_i - \left\langle \tl{\theta} , \vec{x}_i \right\rangle - \left\langle p^{-1} \left( \vec{c} - \tl{\theta} \right) , \vec{x}_i \right\rangle \right) \\
& = & p^{-1} \left( y_i - \left\langle \vec{c} - \vec{x}_i \right\rangle \right),
\ee
and hence $(E-1,I_1,X,Y_1,\vec{c}_1)$ satisfies the assumption of $(E,I,X,Y,\vec{c})$ as long as $E > 1$ and the affine hull of $\set{(\vec{x}_i \bmod p, y_{1,i} \bmod p)}{i \in I_{E-1}}$ in $\Fp^D \times \Fp$ coincides with $\ol{V}$. In particular, replacing $(E,I,X,Y,\vec{c})$ by $(E-1,I_1,X,Y_1,\vec{c}_1)$, we obtain an estimation of the second last digits of entries of $\vec{c}$.

\vs
By recursively replacing $(E,I,X,Y,\vec{c})$ by $(E-1,I_1,X,Y_1,\vec{c}_1)$ under the assumptions on the affine hulls for all recursive steps, we can estimate the trailing $E$-digits of entries of $\vec{c}$ by repetition of Algorithm \ref{last digit linear regression}. Here is a pseudocode of this process:

\begin{figure}[H]
\begin{algorithm}[H]
\caption{Estimation of $\vec{c} \bmod p^E$ from $(E,I,X,Y)$}
\label{trailing digits linear regression}
\begin{algorithmic}[1]
\Function {TrailingDigitsLinearRegression}{$p,E,I,X,Y,\textrm{rep}$}
	\State $\vec{c} \gets (0)_{d=0}^{D}$
	\ForAll {$e \in \N \cap [0,E)$}
		\State $\tl{\theta} \gets$ \Call{LastDigitLinearRegression}{$p,I,X,p^{-e}( y_i - \langle \vec{c} , \vec{x}_i \rangle )_{i \in I},\textrm{rep}$}
		\State $\vec{c} \gets \vec{c} + p^e \tl{\theta}$
		\State $I \gets \set{i \in I}{y_i - \langle \vec{c} , \vec{x}_i \rangle \in p^{e+1} \Zp}$
	\EndFor
	\State \Return $\vec{c}$
\EndFunction
\end{algorithmic}
\end{algorithm}
\end{figure}

We note that when we deal with $p$-adic integers in implementation, we just use its reduction modulo a sufficiently large power $q$ of $p$. We can choose $q$ arithmetically in a way much simpler than the interval arithmetic of real numbers, because of the non-Archimedean property. For example, in Algorithm \ref{trailing digits linear regression}, we only need to consider $X \bmod p^E$ and $Y \bmod p^E$ instead of $X$ and $Y$. The reader should be careful that $X$ cannot be simply replaced by $\ol{X}$ because we used the data of sample points modulo higher power of $p$ in the recursion process.

\vspace{0.3in}
\addcontentsline{toc}{section}{Acknowledgements}
\noindent {\Large \bf Acknowledgements}
\vspace{0.2in}

\noindent
I thank an anonymous referee assigned to the first version of the paper for giving many constructive suggestions such as
\bi
\item[(1)] the suggestion to make the approximate estimations of $\# I_W / \# I$ more tight,
\item[(2)] the suggestion to remove the fitting/verification bias in line 17 in Algorithm \ref{noise-free matrix},
\item[(3)] the suggestion to include results on a specific probabilistic model, which results in inserting \S \ref{Analysis of Probability}, and
\item[(4)] the suggestion to remove the ambiguity on whether $A$ is updated or not through the call of Algorithm \ref{extending indices 2} in Algorithm \ref{linear regression mod p} by returning $(I',A,L)$ instead of $(I',L)$,
\ei
and also for pointing out errors such as
\bi
\item[(5)] the inconsistency of the written pseudocodes with the actual explanations of Algorithm \ref{noise-free matrix} and Algorithm \ref{extending indices 1},
\item[(6)] an arithmetic error in Algorithm \ref{trailing digits linear regression}, and
\item[(7)] an unnatural and wrong assumption on $I$ in \S \ref{Estimation of the Last Digit}
\ei
I thank all people who helped me to learn mathematics and programming. I also thank my family.

%

\end{document}